\documentclass[runningheads,envcountsame,a4paper]{llncs}
\usepackage[english]{babel}

\usepackage[%
rm={oldstyle=false,proportional=true},%
sf={oldstyle=false,proportional=true},%
tt={oldstyle=false,proportional=true,variable=true},%
qt=false%
]{cfr-lm}

\usepackage[dvipdfmx]{graphicx}
\usepackage{tikz}
\usepackage{latexsym}

\usepackage{float}
\usepackage{amsmath}
\usepackage{latexsym}
\usepackage{url}
\usepackage[linesnumbered,ruled]{algorithm2e}
\usepackage{mathtools}
\usepackage{listings}
\usepackage{framed}
\usepackage{esvect}
\usepackage{booktabs}

\definecolor{darkblue}{rgb}{0.0,0,0.5} 
\definecolor{byzantium}{rgb}{0.44, 0.16, 0.39}

\usepackage{pegcommand}

\newcommand{\lpeg}{LPEG}
\newcommand{\GL}{G}
\newcommand{\GP}{G}
\newcommand{\pe}{\Xe}
\newcommand{\re}{r}

\newcommand{\PFtwo}{C_n}
\newcommand{\PFG}{C_G}
\newcommand{\copyfunc}{copy}
\newcommand{\D}{D}
\newcommand{\nfree}{p}
\newcommand{\B}{B}

\begin{document}

\title{Linear Parsing Expression Grammars: An Equivalent Class of DFAs}


\title{Linear Parsing Expression Grammars}

\author{Nariyoshi Chida\inst{1} \and Kimio Kuramitsu\inst{2}}

\tocauthor{Nariyoshi Chida and Kimio Kuramitsu}

\institute{
Yokohama National University, Japan
\email{nariyoshi-chida-pg@ynu.jp}
\and 
Yokohama National University, Japan
\email{kimio@ynu.ac.jp}
}
                        
\maketitle \setcounter{footnote}{0}

\begin{abstract}
PEGs were formalized by Ford in 2004, and have several pragmatic operators (such as ordered choice and unlimited lookahead) for better expressing modern programming language syntax. Since these operators are not explicitly defined in the classic formal language theory, it is significant and still challenging to argue PEGs' expressiveness in the context of formal language theory.
Since PEGs are relatively new, there are several unsolved problems.
One of the problems is revealing a subclass of PEGs that is equivalent to DFAs.
This allows application of some techniques from the theory of regular grammar to PEGs.
In this paper, we define Linear PEGs (LPEGs), a subclass of PEGs that is equivalent to DFAs. Surprisingly, LPEGs are formalized by only excluding some patterns of recursive nonterminal in PEGs, and include the full set of prioritized choice, unlimited lookahead, and greedy repetition, which are characteristic of PEGs. Although the conversion judgement of parsing expressions into DFAs is undecidable in general, the formalism of LPEGs allows for a syntactical judgement of parsing expressions.
\end{abstract}


\section{Introduction}

{\em Deterministic finite automata} (DFAs) are a simple and fundamental theory in the classic formal language, which allows pattern matching on the input without backtracking. 
This positive aspect is applied to the implementation of many regular expression engines such as Google RE2\cite{re2} and grep leading to significantly improved performance.

Similarly, the DFA nature is used for faster parsing. 
For example, a partial conversion of {\em context-free grammars} (CFGs) into DFAs is studied with ANTLR3/4 by Parr {\it et al.}\cite{Parr:2011:LFA:1993498.1993548}\cite{Parr:2014:ALP:2714064.2660202}.
In this study, Parr {\it et al.} achieve better performance of a parser based on CFG by using the conversion.
Concretely, the parser decides a nonterminal that should be expanded by using the DFA.
That is,  DFA conversions remove backtracking while parsing.

In this way, DFAs are used for faster parsing. 
To the best of our knowledge, however, DFAs are not used for parsing a {\em parsing expression grammar} (PEG)\cite{POPL04_PEG} yet.
PEGs are a relatively new and popular foundation for describing syntax, formalized by Ford in 2004.
PEGs look very similar to some of the EBNFs or CFG-based grammar specifications, but differ significantly in that they have {\em unlimited lookahead} with syntactic predicates and deterministic behaviors with greedy repetition and prioritized choice. Due to these extended operators,  PEGs can recognize highly nested languages such as  $\{a^n$ $b^n$ $c^n$ $|$ $n$ $>$ $0\}$, which is not possible in a CFG. 

These extended operators raise an interesting and open question on the connection to the formal language theory. 
In particular, we have expected that a partial DFA conversion brings better performance benefits to the PEG-based parser generation as well as Parr {\it et al.}.
However, parsing expressions are obviously more expressive than DFAs, due to recursion which does not appear in regular expressions. 
Therefore, we require a subclass of PEGs that is equivalent to DFAs for applying DFA techniques to PEGs.

The main contribution of this paper is that we reveal a subclass of PEGs that is equivalent to DFAs.
We formalize the subclass as {\em linear parsing expression grammars} (LPEGs).
Surprisingly, LPEGs are formalized by excluding only some patterns of recursive nonterminal in PEGs, and include the full set of prioritized choice, unlimited lookahead, and greedy repetition, which are unique to PEGs. 
Furthermore, the formalism of LPEGs allows a partial conversion of a PEG into DFAs.
Since converting into DFAs can eliminate backtracking, the partial conversion would lead to further optimization of the parser generator. 

The rest of this paper proceeds as follows. 
Section 2 describes the formalism of LPEGs and shows the relationship between LPEGs and PEGs. 
Section 3 shows a regularity of LPEGs.
Section 4 briefly reviews related work. 
Section 5 is the conclusion. 

\section{Linear PEG}
In this section, we describe the formalism of {\em linear parsing expression grammars} (LPEGs).
LPEGs are a subclass of PEGs equivalent to DFAs, and LPEGs are formalized by excluding patterns of recursive nonterminals that are followed by expressions.
By the exclusion, the syntax of an LPEG is limited to right-linear.
Thus, we can simply consider an LPEG as a PEG where the syntax is right-linear.

To begin with, we describe PEG operators in Section \ref{sec:pegop}.
Then, we show the formalism of LPEGs in Section \ref{sec:lpegdef}.
Finally, we describe language properties in Section \ref{sec:lang}.

\subsection{PEG operators}
\label{sec:pegop}
Table \ref{table:PEGsl} shows the summary of PEG operators used throughout this paper.

\begin{table}[hbt]
\begin{center}
\caption{PEG Operators} \label{table:PEGsl}
\begin{tabular}{llll} \hline
PEG  & Type & Prec. & Description\\ \hline
\verb|' '| & Primary & 5 & Matches text\\
$[ ]$ & Primary & 5 & Matches character class \\
$\Pany$ & Primary & 5 & Any character\\
$A$ & Primary & 5 & Non-terminal application\\
$( {\pe} )$ & Primary & 5 & Grouping\\
${\pe}\Popt$ & Unary suffix & 4 & Option\\
${\pe}\Pzero$ & Unary suffix & 4 & Zero-or-more repetitions\\
${\pe}\Pone$ & Unary suffix & 4 & One-or-more repetitions\\
$\Pand{\pe}$ & Unary prefix & 3 & And-predicate\\
$\Pnot{\pe}$ & Unary prefix & 3 & Not-predicate\\
${\pe}_1 {\pe}_2$ & Binary & 2 & Sequence\\
${\pe}_1 \Por {\pe}_2$ & Binary & 1 & Prioritized Choice\\ \hline
\end{tabular}
\end{center}
\end{table}

The string 'abc' exactly matches the same input, while [abc] matches one of these terminals.
The $\Pany$ operator matches any single terminal. 
The $\pe\Popt$, $\pe\Pzero$, and $\pe\Pone$ expressions behave as in common regular expressions, except that they are greedy and match until the longest position. 
The $\pe_1\;\pe_2$ attempts two expressions $\pe_1$ and $\pe_2$ sequentially, backtracking the starting position if either expression fails.  
The choice $\pe_1\Por \pe_2$ first attempts $\pe_1$ and then attempts $\pe_2$ if $\pe_1$ fails. 
The expression $\Pand\pe$ attempts $\pe$ without any terminal consuming. 
The expression $\Pnot\pe$ fails if $\pe$ succeeds, but succeeds if $\pe$ fails. 

We consider the any character $\Pany$ expression to be a choice of all single terminals ($a\Por b\Por ...\Por c$) in $\Sigma$. 
As long as any special cases are not noted, we treat the any character as a syntax sugar of such a terminal choice.

Likewise, many convenient notations used in PEGs such as character class, option, one or more repetition, and and-predicate are treated as syntax sugars:

\begin{equation}
\begin{array}{llll}
[abc]    & =    &  a\Por b\Por c & \mbox{character class} \\
\pe\Popt   & =   & \pe\Por \Pempty  & \mbox{option} \\
\pe^{\Pone}    & =    &  \pe \pe^{\Pzero} & \mbox{one or more repetition} \\
\Pand\pe   & =   & \Pnot\Pnot\pe  & \mbox{and-predicate} \\ 
\end{array}
\notag
\end{equation}

Furthermore, we can eliminate zero or more repetition from a PEG by using a new nonterminal.
\begin{eqnarray*}
	\pe\Pzero & = & A \leftarrow e A \Por \epsilon \\
\end{eqnarray*}

\subsection{Definition of LPEGs}
\label{sec:lpegdef}

\begin{definition}
\label{def:LPEG}
A linear parsing expression grammar ({\lpeg}) is defined by a 4-tuple ${\GL} = (N_{\GL}, \Sigma, P_{\GL}, \pe_s)$, where $N_{\GL}$ is a finite set of nonterminals, $\Sigma$ is a finite set of terminals, $P_{\GL}$ is a finite set of production rules, and ${\pe}_s$ is a linear parsing expression termed the start expression.
A linear parsing expression $\pe$ is a parsing expression with the syntax according to BNF shown in Fig. \ref{fig:lpegsyntax}.
$\nfree$ in Fig. \ref{fig:lpegsyntax} is a nonterminal-free parsing expression (n-free parsing expression).
An n-free parsing expression $\nfree$ is a parsing expression such that the expression doesn't contain nonterminals.
Each rule in $P_{\GL}$ is a mapping from a nonterminal $A \in N_{\GL}$ to a linear parsing expression ${\pe}$.
We write $P_{\GL}(A)$ to denote an associated expression ${\pe}$ such that $A \leftarrow {\pe} \in P_{\GL}$.

\begin{figure}[htb]
\begin{equation}
\begin{array}{llll}
{\pe} &::=& \nfree \\
  &|&     \nfree\ A \\
  &|&     \nfree\ \pe \\
  &|&	    {\pe}\Por{\pe} \\
  &|&    \Pnot {\pe}\ {\pe} \\
\nfree    & ::=    &  \Pempty& \mbox{empty} \\
   & |   & a & \mbox{character} \\ 
   & |   & \Pany & \mbox{any character}\\
   & |   & \nfree\ \nfree& \mbox{sequence}\\
   & |   & \nfree \Por \nfree& \mbox{prioritized choice} \\
   & |   & \nfree\Pzero& \mbox{zero or more repetition} \\
   & |   & \Pnot \nfree& \mbox{not-predicate} \\
\end{array}
\notag
\end{equation}
\caption{Syntax of a linear parsing expression}
\label{fig:lpegsyntax}
\end{figure}
\end{definition}

We show two examples of an LPEG and an example of a PEG but not an LPEG.

\begin{example}
${\GL} = (\{A,B\}, \{a,b,c\}, \{A \leftarrow a A \Por b B \Por c, B \leftarrow a B \Por b A \Por c\}, A)$ is an LPEG.
\end{example}
\begin{example}
${\GL} = (\{A\}, \{a,b\}, \{A \leftarrow \Pnot(aA) aA \Por b\}, A)$ is an LPEG.
\end{example}
\begin{example}
${\GL} = (\{A,B\}, \{a,b\}, \{A \leftarrow a A a \Por B\Pzero, B \leftarrow a B \Por b\}, A)$ is not an LPEG.
Note that $a A a$ and $B\Pzero$ are not derived from the above syntax.
\end{example}

All subsequent use of the unqualified term ``grammar'' refers specifically to linear parsing expression grammars as defined here,
and the unqualified term ``expression'' refers to linear parsing expressions.
We use the variables $a,b,c,d \in \Sigma$, $A,B \in N_{\GL}$, $w, x, y, z \in \Sigma^*$, and $\pe$ for linear parsing expressions.

\subsection{Language Properties}
\label{sec:lang}
In this section, we define a language recognized by LPEGs.
We use a function $consume$ to define the language.
The definition of the function $consume$ is as follows.
\begin{itemize}
\item $consume(\pe,x) = y$ denotes that the expression $\pe$ succeeds on the input string $x$ and consumes $y$.
\item $consume(\pe,x) = fail$ denotes that the expression $\pe$ fails on the input string $x$.
\end{itemize}

\begin{definition}
\label{def:lpeglang1}
Let ${\GL} = (N_{\GL}, \Sigma, P_{\GL}, \pe_s)$ be an LPEG, let $\pe$ be an expression.
The language generated by $\pe$ is a set of all strings over $\Sigma$:

$L_{\GL}({\pe}) = \{ x \mid  x \in \Sigma^{*}, \mbox{$consume$($\pe$,$x$) = $y$}\}$.
\end{definition}

\begin{definition}
\label{def:lpeglang2}
Let ${\GL} = (N_{\GL}, \Sigma, P_{\GL}, \pe_s)$ be an LPEG.
The language generated by a grammar ${\GL}$ is a set of all strings over $\Sigma$:

$L({\GL}) = L_{{\GL}}(\pe_s)$.
\end{definition}

We define that two parsing expressions are {\em equivalent} as follows:
\begin{definition}
Let ${\GP}_1 = (N_{{\GP}_1},\Sigma,P_{{\GP}_1},{\pe_s}_1)$ and ${\GP}_2 = (N_{{\GP}_2},\Sigma,P_{{\GP}_2},{\pe_s}_2)$.
Two parsing expressions ${\pe_s}_1$ and ${\pe_s}_2$ are equivalent if $consume({\pe_s}_1,x) = consume({\pe_s}_2,x)$ for any input string $x \in \Sigma^{*}$.
\end{definition}

If a parsing expression $\pe_1$ and $\pe_2$ are equivalent, we can rewrite $\pe_1$ as $\pe_2$, and vice versa since the languages are same.
\begin{theorem}
Let ${\GP}_1 = (N_{{\GP}_1},\Sigma,P_{{\GP}_1},{\pe_s}_1)$ and ${\GP}_2 = (N_{{\GP}_2},\Sigma,P_{{\GP}_2},{\pe_s}_2)$.
If the parsing expression ${\pe_s}_1$ and ${\pe_s}_2$ are equivalent, then $L({\GP}_1) = L({\GP}_2)$.
\end{theorem}
\begin{proof}
Trivial, from the definition of the language.
\end{proof}

\section{Regularity}
In this section, we prove that LPEGs are a class that is equivalent to DFAs.
To prove this, we show that for any LPEG ${\GL}$ there exists a DFA ${\D}$ such that $L(\GL) = L(\D)$ and for any DFA $\D$ there exists an LPEG $\GL$ such that $L(\D) = L(\GL)$.
We show the former in section \ref{sec:fltd} and the latter in section \ref{sec:fdtl}.

\subsection{From LPEGs to DFAs}
\label{sec:lpegdfa}
\label{sec:fltd}

We show that for any LPEG ${\GL}$ there exists a DFA $D$ such that $L({\GL}) = L(D)$.
This can be proved by translating LPEGs into {\em boolean finite automata}(BFAs)\cite{BRZOZOWSKI198019}.

A BFA is a generalized {\em nondeterministic finite automaton} (NFA).
The difference between NFAs and BFAs is a representation of a state under transition.
The state under transition on NFAs can be represented as a boolean function consisting of logical OR and boolean variables.
On the other hand, the state under transition on BFAs can be represented as a boolean function consisting of logical AND, logical OR, logical NOT, constant values (i.e. $true$ and $false$), and boolean variables.

There are two reasons for using BFAs.
One is to handle not-predicates.
We can represent these predicates as a boolean function by using logical AND and logical NOT.
Another reason is that BFAs can be converted into DFAs (\cite{BRZOZOWSKI198019}, Theorem 2).
Thus, LPEGs can be converted into DFAs if we can convert LPEGs into BFAs.

In the next section we describe basic definitions and notations of BFAs.
In order to make the conversion easier, sets of the accepting states of BFAs are divided into two sets in the definitions.
In Section \ref{sec:lpegbfa}, we show that LPEGs can be converted into BFAs.

\subsubsection{Boolean Finite Automata}
\begin{definition}
A {\em boolean finite automaton} (BFA) is a 6-tuple $\B$=$(Q,\Sigma,\delta,f^0,F,P)$.
$Q = \{q_1,q_2,...,q_n\}$ is a finite set of states.
$\Sigma$ is a finite set of terminals.
$\delta : Q \times \Sigma \rightarrow V_Q$ is a transition function that maps a state and a terminal into a boolean function of boolean variables that correspond to the states $q_1, q_2, ..., q_n$ in the set of boolean functions $V_Q$.
$f^0 \in V_Q$ is an initial boolean function.
$F$ is a finite set of accepting states.
$P$ is a finite set of accepting states for lookaheads. 
We also use $q_i$ as a boolean variable that corresponds to a state $q_i \in Q$.
\end{definition}

Let $f$ be a boolean function in $V_Q$.
The transition function $\delta$ is extended to $V_Q \times \Sigma^*$ as follows:
\begin{eqnarray*}
	\delta(f,\epsilon) &=& f\\
	\delta(f,a) &=& f(\delta(q_1,a),...,\delta(q_n,a))\\
	\delta(f,aw) &=& \delta(\delta(f,a),w)
\end{eqnarray*}

A language accepted by a BFA is defined as follows:
\begin{definition}
Let $B$ be a BFA and $x \in \Sigma^*$.
$x \in L(\B)$ iff $\delta(f^0,x)(c_1,...,c_n) = true$, where $c_i = true$ if $q_i \in F\cup P$, otherwise $false$.
\end{definition}

Then, we define a function $consume$ for a BFA $B$ in the same way as the function $consume$ for an LPEG.
We use the function to show the equivalence between the language of an LPEG and a BFA converted from the LPEG.
The function $consume$ is defined as follows.
\begin{itemize}
\item $consume(B,w) = x$ denotes that $eval_P(\delta(eval_F(\delta(f^0,x),F),y),P) = true$ for an input string $w=xyz$.
\item $consume(B,w) = fail$ if there is no such $x$.
\end{itemize}
In  the definition, we use two evaluation functions, $eval_F$ and $eval_P$.
These functions are used for representing the behaviors of the predicate operators  of LPEGs on BFAs.
Specifically, $eval_F$ takes a boolean function $f$ and a set of accepting state $F$ and returns a boolean function $f'$ that replaced boolean variables $q_i\in F$ in $f$ with $true$.
For example, let $f=q_0\land(q_1\lor q_2)$ and $F=\{q_1\}$.
Then $eval_F(f,F) = q_0\land(true\lor q_2) = q_0\land q_2$.
$eval_P$ takes a boolean function $f$ and a set of accepting state of not-predicates $P$ and returns a boolean value that is a result of a replacement of a boolean variables $q_i \in P$ in $f$ with $true$, otherwise $false$.
For example, let $f=q_0\land\overline{q_1}$ and $P=\{q_0\}$.
Then, $eval_P(f,P) = true\land\overline{false} = true$.
To understand the intuition of the functions, we show an example in Fig. \ref{fig:eval_int}.
In the figure, we show an LPEG such that the start expression $\pe_S = \Pand(ab)a$ and a BFA $B$ that the language is equivalent to the language of the LPEG.
The circles, double circles and arrows denote the states, accepting states and transitions of the BFA, respectively.
The arrow labeled $\&$ denotes the and-predicate operator in LPEGs.
\begin{figure}[H]
  \begin{center} 
    \includegraphics[width=7.0cm]{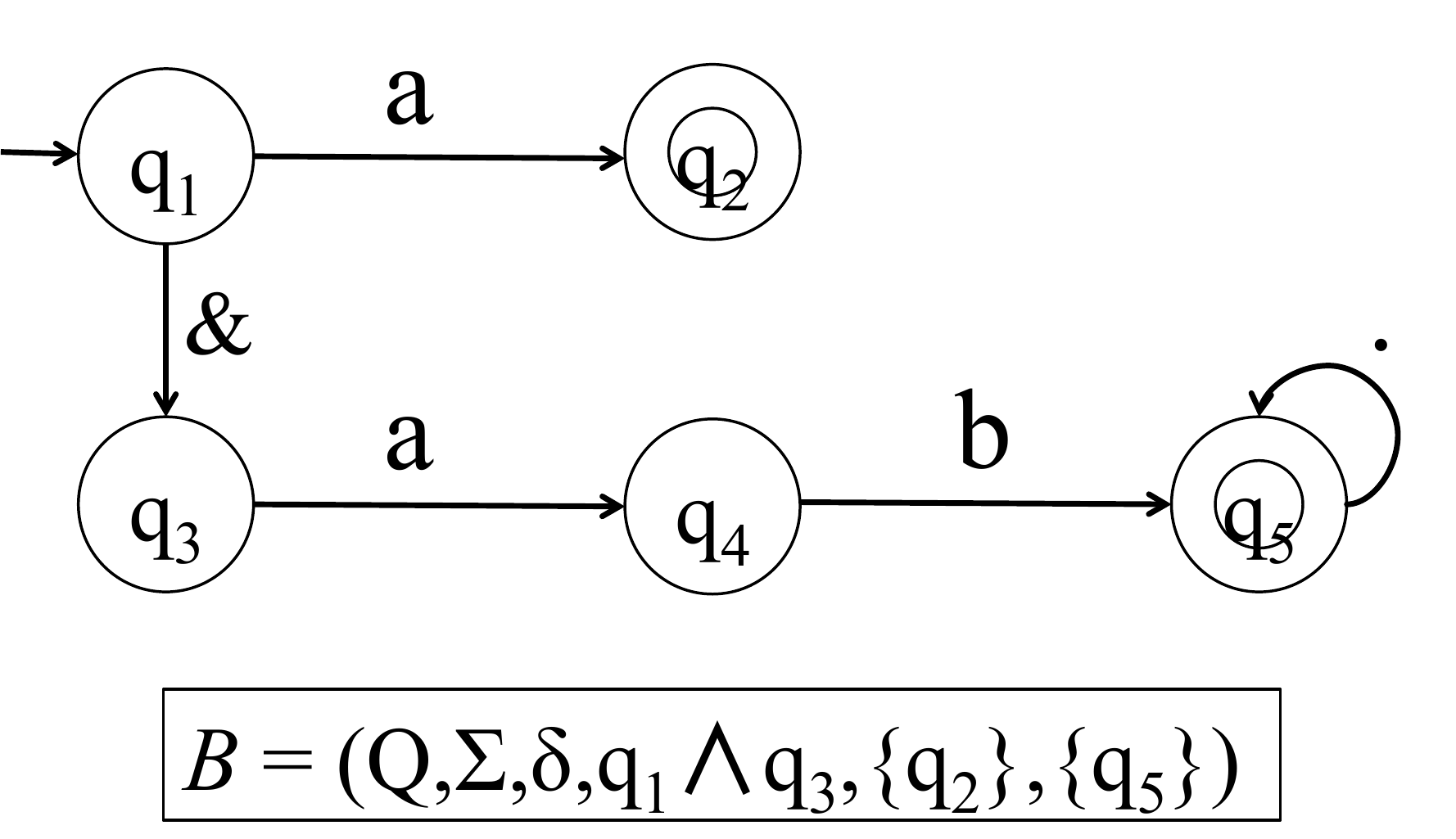}
    \caption{The BFA that the language is equivalent to the language of the LPEG} 
    \label{fig:eval_int}
  \end{center}
\end{figure}
Let $ab$ be an input string for the BFA $B$.
Then, $consume(B,ab)$ is evaluated as follows.
\begin{eqnarray*}
	consume(B,ab) &=& eval_P(\delta(eval_F(\delta(q_1\land q_3,a),\{q_2\}),b),\{q_5\})\\
	&=&  eval_P(\delta(eval_F(q_2\land q_4,\{q_2\}),b),\{q_5\})\\
	&=&  eval_P(\delta(true\land q_4,b),\{q_5\})\\
	&=&  eval_P(q_5,\{q_5\})\\
	&=& true
\end{eqnarray*}
This result shows that the BFA $B$ accepts the input string $ab$.
In this way, $eval_F$ evaluates the operators other than predicate operators and $eval_P$ evaluates the predicate operators.

We define that an LPEG $G$ and a BFA $B$ are {\em equivalent} as follows:
\begin{definition}
Let $G$ and $B$ be an LPEG such that the start expression is $\pe_S$ and a BFA, respectively.
$G$ and $B$ are equivalent if $consume(\pe_S,x) = consume(B,x)$ for any input string $x \in \Sigma*$.
\end{definition}

\begin{theorem}
\label{theo:lbeq}
Let $G$ and $B$ be an LPEG and a BFA, respectively.
If $G$ and $B$ are equivalent, then $L(G) = L(B)$.
\end{theorem}
\begin{proof}
We can prove this by case analysis of the result of the function $consume$.
\end{proof}

\subsubsection{From LPEGs to BFAs}
\label{sec:lpegbfa}
We show a conversion from an LPEG into a BFA.
The conversion consists of four steps.
In the first step, we rewrite a prioritized choice with an alternation in regular expressions.
In the second step, we add new production rules for nonterminals in not-predicates to an LPEG.
We apply these two steps in order to simplify the conversion.
In the third step, we convert a modified LPEG into a BFA.
However, the BFA is incomplete in this step, since the conversion handles nonterminals as temporary boolean variables to avoid an infinite loop by recursions.
In the final step, we replace the temporary boolean variables in a BFA with initial functions of the nonterminals.

First, we rewrite a prioritized choice $\Por$ with an alternation $\REor$ in regular expressions.
That is, we rewrite $\pe_1 \Por \pe_2$ with $\pe_1 \Por \Pnot\pe_1 \pe_2$.
We show an example of the rewriting as follows:
\begin{example}
${\GL} = (\{A,B\},\{a,b,c,d\},\{A \leftarrow \Pnot(a A \Por b B \Por c ) d, B \leftarrow a B \Por b\},A)$ is an LPEG.
Then, we rewrite the LPEG as $(\{A,B\},\{a,b,c,d\},\{A \leftarrow \Pnot(a A \REor \Pnot(a A) b B \REor \Pnot(b B) c ) d, B \leftarrow a B \REor \Pnot(a B) b\},A)$
\end{example}

By the modification, we can apply Thompson's construction\cite{Thompson:1968:PTR:363347.363387} to the construction of a BFA of a prioritized choice.
We show that the language is the same before and after the rewriting in Theorem \ref{theo:orderunorder}.

\begin{theorem}
\label{theo:orderunorder}
$\pe_1\Por\pe_2$ is equivalent to $\pe_1 \mid \Pnot\pe_1\pe_2$.
\end{theorem}
\begin{proof}
We can prove this by case analysis.
\begin{description}	
\item[1.]\mbox{Case $consume(\pe_1,w) = x$ and $consume(\pe_2,w) = x'$}\\
	By the semantics of a prioritized choice, $consume(\pe_1\Por\pe_2,w) = x$.
	Since $consume(\pe_1,w) = x$ and $consume(\Pnot\pe_1,w) = fail$, $consume(\pe_1\REor\Pnot\pe_1\pe_2,w) = x$.
	Hence, $consume(\pe_1\Por\pe_2,w) = consume(\pe_1\REor\Pnot\pe_1\pe_2,w)$.
\item[2.]\mbox{Case $consume(\pe_1,w) = x$ and $consume(\pe_2,w) = fail$}\\
	This is the same as the case 1.
\item[3.]\mbox{Case $consume(\pe_1,w) = fail$ and $consume(\pe_2,w) = x$}\\
	By the semantics of a prioritized choice, $consume(\pe_1\Por\pe_2,w) = x$.
	Since $consume(\pe_1,w) = fail$, $consume(\Pnot\pe_1,w) = \epsilon$ and $consume(\pe_2,w) = x$ , $consume(\pe_1\REor\Pnot\pe_1\pe_2,w) = x$.
	Hence, $consume(\pe_1\Por\pe_2,w) = consume(\pe_1\REor\Pnot\pe_1\pe_2,w)$.
\item[4.]\mbox{Case $consume(\pe_1,w) = fail$ and $consume(\pe_2,w) = fail$}\\
	By the semantics of a prioritized choice, $consume(\pe_1\Por\pe_2,w) = fail$.
	Since $consume(\pe_1,w) = fail$, $consume(\Pnot\pe_1,w) = \epsilon$ and $consume(\pe_2,w) = fail$, $consume(\pe_1\REor\Pnot\pe_1\pe_2,w) = fail$.
	Hence, $consume(\pe_1\Por\pe_2,w) = consume(\pe_1\REor\Pnot\pe_1\pe_2,w)$.
\end{description}
\end{proof}

Secondly, we add new production rules for nonterminals in not-predicates to an LPEG.
We apply this modification to LPEGs, because we consider a nonterminal $A$ in a not-predicate and a nonterminal $A$ that is not in a not-predicate as distinct.
We show the modification in Definition \ref{def:PFG}.

\begin{definition}
\label{def:PFG}
Let ${\GL} = (N_{\GL},\Sigma,P_{\GL},\pe_s)$ be an LPEG.
${\PFG}({\GL}) = (N_{\GL}\cup N_{\GL'},\Sigma,\\P_{{\GL}_1}\cup P_{{\GL}_2},\pe_{s'})$, where
$P_{{\GL}_1} = \{A \leftarrow {\PFtwo}(\pe_{A}) \mid \pe_A \in P_{\GL}\}$,
$P_{{\GL}_2} = \{A' \leftarrow \pe_{A'} \mid \pe_{A'} = {\copyfunc}(\pe_A), \pe_A \in P_{\GL}\}$,
$N_{\GL'} = \{A' \mid \pe_{A'} \in P_{{\GL}_2}\}$,
$\pe_{s'} = {\PFtwo}(\pe_s)$.
\end{definition}

In the modification function, we use an auxiliary function $\PFtwo$.
$\PFtwo$ is a function for modification of a production rule.
We show the definition of $\PFtwo$ in Definition \ref{def:PFtwo}.
In the following definition, we use a function $\copyfunc$.
$\copyfunc(\pe) = \pe'$ denotes that a nonterminal $A$ is renamed as $A'$ if the nonterminal $A$ is not already $A'$ and the other expressions are same.
We assume that there does not exist $A'$ in an LPEG before the modification.

For example, $\copyfunc(a A \REor \Pnot(a A) b \REor (\Pnot b) c) = a A' \REor \Pnot(a A') b \REor (\Pnot b) c$ and $\copyfunc(\copyfunc(\Pnot(a A))) = \copyfunc(\Pnot(a A')) = \Pnot(a A')$.
\begin{definition}
\label{def:PFtwo}
\begin{equation}
\begin{array}{llll}
{\PFtwo}(\nfree) &=& \nfree\\
{\PFtwo}(\nfree\ A)  &=&     \nfree\ A\\
{\PFtwo}(\nfree\ {\pe})  &=&     \nfree\ {\PFtwo}({\pe})\\
{\PFtwo}({\pe}\REor{\pe})  &=&	    {\PFtwo}({\pe})\REor{\PFtwo}({\pe})\\
{\PFtwo}(\Pnot {\pe}\ {\pe})  &=&    \Pnot ({\copyfunc}({\pe}))\ {\PFtwo}({\pe})\\
\end{array}
\notag
\end{equation}
\end{definition}

We show an example of the modification as follows:

\begin{example}
${\GL} = (\{A,B\},\{a,b,c,d\},\{A \leftarrow \Pnot(a A \REor \Pnot(a A) b B \REor \Pnot(b B) c ) d, B \leftarrow a B \REor \Pnot(a B) b\},A)$ is an LPEG.
Then, ${\PFG}({\GL}) = (\{A,B,A',B'\},\{a,b,c,d\},P_{{\GL}_1}\cup P_{{\GL}_2},A)$, where $P_{{\GL}_1}$ consists of the following rules:
\begin{eqnarray*}
	A &\leftarrow& \Pnot(a A' \REor \Pnot(a A') b B' \REor \Pnot(b B') c ) d\\
	B &\leftarrow& a B \REor \Pnot(a B') b
\end{eqnarray*}
$P_{{\GL}_2}$ consists of the following rules:
\begin{eqnarray*}
	A' &\leftarrow& \Pnot(a A' \REor \Pnot(a A') b B' \REor \Pnot(b B') c  ) d \\
	B' &\leftarrow& a B' \REor \Pnot(a B') b 
\end{eqnarray*}
\end{example}

Thirdly, we describe the conversion from modified LPEGs to BFAs with temporary boolean variables.
The foundation of the conversion follows Morihata's work\cite{REwLAtoDFA} for regular expression with positive and negative lookaheads, but we significantly extend his work with handling recursion.

In this function, we assume that the names of boolean variables are distinct in the conversion.
We write a temporary boolean variable of a nonterminal $A$ as $f_{tmp_A}$.
A function $\phi(f_1,f_2,F)$ converts the boolean function $f_1$ by replacing a boolean variable $s$ in $f_1$ with $s \lor f_2$ if $s \in F$.
For example, let $f_1 = (q_1\land q_2)\lor q_3$, $f_2 = q_4$ and $F = \{q_2,q_3\}$, where $q_1$, $q_2$, $q_3$ and $q_4$ are boolean variables.
Then, $\phi(f_1,f_2,F) = (q_1\land ( q_2 \lor q_4 ) ) \lor ( q_3 \lor q_4 )$.
Note that the BFA converted by the following function accepts the full match of the expressions.
Therefore, a BFA that accepts the same language with the LPEG is written as $T(\pe_s\Pany\Pzero)$.

\begin{eqnarray*}\nonumber
	T(\epsilon) &=& (\{s\}, \Sigma, \{\}, s, \{s\}, \{\})\nonumber\\
	T(a) &=& (\{s,t\}, \Sigma, \{((s,a),t)\}, s,\{t\}, \{\})\nonumber\\
	T(\Pnot\pe) &=& (Q \cup \{s\}, \Sigma, \delta\cup\{((t,\Pany),t)\mid t\in F\cup P\}, s\land\overline{f^0}, \{s\}, F \cup P)\nonumber\\
		               &where& (Q, \Sigma, \delta, f^0,F, P) = T(\pe)\\
	T(\pe_1 \pe_2) &=& (Q_1 \cup Q_2, \Sigma, \delta, \phi(f^0_1,f^0_2,F_1), F_2, P_1 \cup P_2)\nonumber\\
				              			            &where& (Q_1, \Sigma, \delta_1, f^0_1, F_1, P_1) = T(\pe_1),\\
							            	&& (Q_2, \Sigma, \delta_2, f^0_2, F_2, P_2) = T(\pe_2)\\
							                    &and&            \delta=  \{ ((s,a),\phi(t,f^0_2,F_1)) \mid ((s,a),t) \in \delta_1 \} \cup\delta_2\nonumber\\
	T(\pe_1 \REor \Pnot\pe_1 \pe_2) &=& (Q_1\cup Q_2, \Sigma, \delta_1 \cup \delta_2,f^0_1 \lor f^0_2,F_1\cup F_2, P_1 \cup P_2)\\
	                                     &where& (Q_1,\Sigma,\delta_1, f^0_1, F_1, P_1) = T(\pe_1)\\
	                                     &and& (Q_2,\Sigma,\delta_2, f^0_2, F_2, P_2) = T(\Pnot\pe_1\pe_2)\\
	T(A) &=&\left\{ \begin{array}{ll}
	    T(P_G(A))\;\; \verb|(first application)|\\
	    (\{\},\Sigma,\{\},f_{tmp_A},\{\},\{\})\;\;  \verb|(otherwise)|\\
	  \end{array} \right.
\nonumber\end{eqnarray*}

The function $T$ handles the nonterminals in the same way as a conversion from a right-linear grammar to an NFA\cite{Linz:2006:IFL:1206580}.
In the conversion from a right-linear grammar to an NFA, a nonterminal is handled as an initial state of the NFA.
In the same way, in the function $T$, a nonterminal is handled as an initial function of the BFA.

Finally, we replace temporary variables with the initial functions of the nonterminals.
We show an example of conversion from an LPEG to a BFA.

\begin{example}
Let $G = (\{A\},\{a,b\},\{A \leftarrow a A \Por b\},A)$ be an LPEG.
The language of the LPEG $L(G) = \{ a^ib \mid i \geq 0 \}$.

First, we rewrite the prioritized choice in the LPEG as follows.
$G = (\{A\},\{a,b\},\{A \leftarrow a A \REor \Pnot(a A) b\},A)$.

Second,y we modify the LPEG $G$ as follows.
$G = \PFG(G)$, where $\PFG(G) = (\{A,A'\},\{a,b\},\{A \leftarrow a A \REor \Pnot(a A') b, A' \leftarrow a A' \REor \Pnot(a A') b\},A)$.

Thirdly, we convert the LPEG $G$ to a BFA $B$ with temporary boolean variables.
As a result of the conversion, we get the BFA $B = (\{q_0,...,q_{13}\},\{a,b\},\delta,q_0\lor((q_{11}\lor q_{12})\land\overline{q_2}),\{q_{13}\},\{q_{10}\})$, where $\delta$ is shown in Table \ref{tab:delta_tmp}.
For simplicity, we consider transitions that are not in Table \ref{tab:delta_tmp} return $false$.

\begin{table}[htb]
\begin{center}
\caption{The transition function $\delta$ with temporary boolean variables}
\label{tab:delta_tmp}
  \begin{tabular}{ccc} \toprule
    state$\backslash$terminal & a & b  \\ \midrule 
    $q_0$ & $q_1\lor f_{tmp_A}$ & $false$  \\
     $q_2$ & $q_3\lor q_4\lor((q_8\lor q_9)\land\overline{q_6})$ & $false$ \\
    $q_4$ & $q_5\lor f_{tmp_{A'}}$ & $false$ \\ 
    $q_6$ & $q_7\lor f_{tmp_{A'}}$ & $false$ \\
    $q_9$ & $false$ & $q_{10}$\\
    $q_{10}$ & $q_{10}$ & $q_{10}$ \\
    $q_{12}$ & $false$ & $q_{13}$ \\ \bottomrule
  \end{tabular}
  \end{center}
\end{table}

Finally, we replace temporary boolean variables with the initial functions.
In this BFA, there are two temporary boolean variables, $f_{tmp_{A}}$ and $f_{tmp_{A'}}$.
$f_{tmp_{A}}$ is replaced by $q_0\lor((q_{11}\lor q_{12})\land\overline{q_2})$.
$f_{tmp_{A'}}$ is replaced by $q_4\lor((q_8\lor q_9)\land\overline{q_6})$.
The transition function $\delta$ is shown in Table \ref{tab:delta}.

\begin{table}[htb]
\begin{center}
\caption{The transition function $\delta$}
\label{tab:delta}
  \begin{tabular}{ccc} \toprule
    state$\backslash$terminal & a & b  \\ \midrule 
    $q_0$ & $q_1\lor {q_0\lor((q_{11}\lor q_{12})\land\overline{q_2})}$ & $false$  \\
     $q_2$ & \shortstack{$q_3\lor q_4\lor((q_8\lor q_9)\land$$\overline{q_6})$} & $false$ \\
    $q_4$ & $q_5\lor q_4\lor((q_8\lor q_9)\land\overline{q_6})$ & $false$ \\ 
    $q_6$ & $q_7\lor q_4\lor((q_8\lor q_9)\land\overline{q_6})$ & $false$ \\
    $q_9$ & $false$ & $q_{10}$\\
    $q_{10}$ & $q_{10}$ & $q_{10}$ \\
    $q_{12}$ & $false$ & $q_{13}$ \\ \bottomrule
  \end{tabular}
  \end{center}
\end{table}

The BFA $B$ accepts an input string $b$.
\begin{eqnarray*}
\delta(q_0\lor((q_{11}\lor q_{12})\land\overline{q_2}),b) &=& false\lor((false\lor q_{13})\land\overline{false})\\
	&=& false\lor((false\lor true)\land\overline{false})\\
	&=& true
\end{eqnarray*}
In the same way, we can check that the BFA $B$ rejects an input string $a$.

\end{example}

\begin{theorem}
\label{theo:lpegtobfa}
Let $G = (N,\Sigma,R,\pe_S)$ be an LPEG modified in Definition \ref{def:PFG}.
Let $B = T(\pe_S\Pany\Pzero)$ and $B$ has already replaced the temporary variables with initial functions.
Then, $L(G) = L(B)$.
\end{theorem}
\begin{proof}
We show that an LPEG $G$ is equivalent to the BFA $B$ by induction on the structure of an linear parsing expression $\pe$.
We assume that $T(\pe)$ is a BFA such that $consume(T(\pe),w) = consume(\pe,w)$, where $w \in \Sigma^*$.

The basis is as follows:
{\footnotesize\bf\flushleft BASIS}:
\begin{description}	
\item[1.]\mbox{Case $\pe = \epsilon$}\\
	The expression $\epsilon$ does not fail to match any string.
	Thus, we only need to consider the case $consume(\epsilon,w) = \epsilon$.
	$consume(T(\epsilon),w) = \epsilon$ since $T(\epsilon) = (\{s\},\Sigma,\{\},s,\{s\},\{\})$ and $eval_P(\delta(eval_F(\delta(s,\epsilon),F),\epsilon),P) = true$.
	Thus, $consume(T(\epsilon),w) = consume(\epsilon,w)$.
\item[2.]\mbox{Case $\pe = a$}\\
	Let $T(a) = (\{s,t\},\Sigma,\{((s,a),t)\},s,\{t\},\{\})$.
	When $consume(a,w) = a$, $consume(T(a),w) = a$ since $eval_P(\delta(eval_F(\delta(s,a),F),\epsilon),P) = true$.
	When $consume(a,w) = fail$, $consume(T(a),w) = fail$ since the first character of $w$ is not $a$ and $eval_P(\delta(eval_F(\delta(s,b),F),\epsilon),P) = false$, where $b$ is a prefix of $w$.
	Thus, $consume(T(a),w) = consume(a,w)$.
\end{description}

The induction is as follows:
{\footnotesize\bf\flushleft INDUCTION}:
\begin{description}	
\item[1.]\mbox{Case $\pe = \Pnot\pe$}\\
By induction hypothesis, $consume(\pe,w) = consume(T(\pe),w)$.
When $consume(\pe,w) = fail$, $consume(\Pnot\pe,w) = \epsilon$.
$consume(T(\Pnot\pe),w) = eval_P(\delta(eval_F(\delta(s\land\overline{f^0},\epsilon),\{s\}),xy),F\cup P)$
$ = eval_P(\delta(\overline{f^0},xy),F\cup P)$.
Since the transition function $\delta$ has $\{((t,\Pany),t) \mid t\in F\cup P\}$, $eval_P(\delta(f,xy),F\cup P) = eval_P(\delta(eval_F(\delta(f,x),F),y),P)$.
Thus, $eval_P(\delta(\overline{f^0},xy),F\cup P) = eval_P(\delta(eval_F(\delta(\overline{f^0},x),F),y),P) = true$.
Hence, $consume(T(\Pnot\pe),w) = \epsilon$.

When $consume(\pe,w) = x$, $consume(\Pnot\pe,w) = fail$.
When $consume(T(\pe),w) = x$, $eval_P(\delta(eval_F(\delta(f^0,x),F),y),P)  = true$ and
$consume(T(\Pnot\pe),w) = fail$ since $eval_P(\delta(eval_F(\delta(s\land\overline{f^0},x),\{s\}),y),F\cup P) = s'\land\overline{true} = false$, where $s'$ is $true$ or $false$.

Hence $consume(\Pnot\pe,w) = consume(T(\Pnot\pe),w)$.

\item[2.]\mbox{Case $\pe = \pe_1\pe_2$}\\
We can divide the case into three cases: $consume(\pe_1,w) = x$ and $consume(\pe_2,yz) = y$, $consume(\pe_1,w) = fail$, and $consume(\pe_1,w) = x$ and $consume(\pe_2,yz) = fail$.
We show the first case.
When $consume(\pe_1,w) = x$ and $consume(\pe_2,yz) = y$, $consume(\pe_1\pe_2,w) = xy$.
$consume(T(\pe_1\pe_2),w) = xy$ since $eval_P(\delta(eval_F(\delta(\phi(f^0_1,f^0_2,F_1),xy),F_2),z_1),P_1\cup P_2)  = true$.
Note that we do not need to consider about predicates because $consume(\pe_1,w) = x$ and $consume(\pe_2,yz) = y$, that is, predicates in $\pe_1$ and $\pe_2$ succeeds on $w$ and $yz$, respectively.
In the same way, we can confirm that $consume(\pe_1\pe_2,w) = consume(T(\pe_1\pe_2),w)$.

\item[3.]\mbox{Case $\pe = \pe_1\REor\Pnot\pe_1\pe_2$}\\
We can divide the case into three cases: $consume(\pe_1,w) = x$, $consume(\pe_1,w) = fail$ and $consume(\pe_2,w) = x$, and $consume(\pe_1,w) = fail$ and $consume(\pe_2,w) = fail$.
Note that there is no case such as $consume(\pe_1,w) = x$ and $consume(\pe_2,w) = y$ since $\pe_2$ does not match the input string $w$ if $\pe_1$ matches $w$.
We show the first case.
When $consume(\pe_1,w) = x$, $consume(\pe_1\REor\Pnot\pe_1\pe_2,w) = x$.
Let $T(\pe_1) = (Q_1,\Sigma,\delta_1,f^0_1,F_1,P_1)$ and $T(\pe_2) = (Q_2,\Sigma,\delta_2,f^0_2,F_2,P_2)$.
In this case, $consume(T(\pe_1),w) = x$ and $eval_P(\delta(eval_F(\delta(f^0_1,x),F_1),y),P_1)  = true$.
In addition, $consume(T(\pe_2),w) = x'$ and $eval_P(\delta(eval_F(\delta(f^0_2,x'),F_2),y'),P_2)  = true$.
The initial function of $T(\Pnot\pe_1\pe_2)$ is $(s\lor f^0_2)\land\overline{f^0_1}$.
$consume(T(\pe_1\REor\Pnot\pe_1\pe_2),w) = x$ since $eval_P(\delta(eval_F(\delta(f^0_1\lor((s\lor f^0_2)\land\overline{f^0_1}),x),F_2),y),P_1\cup P_2)  = true\lor false = true$.
In the same way, we can confirm that $consume(\pe_1\REor\Pnot\pe_1\pe_2,w) = consume(T(\pe_1\REor\Pnot\pe_1\pe_2),w)$.

\item[4.]\mbox{Case $\pe = A$}\\
At the first application, $consume(T(R(A)),w) = consume(A,w)$ for any string $w \in \Sigma^*$ by the assumption.
Otherwise, $consume(T(R(A)),w) = consume(A,w)$ since the boolean function $f_{tmpA}$ is the initial function of the BFA $T(R(A))$ and there already exist other elements of the BFA.
\end{description}

Hence, the LPEG $G$ is equivalent to the BFA $B$.
Thus, by Theorem \ref{theo:lbeq}, $L(G) = L(B)$.
\end{proof}

\begin{theorem}
\label{theo:lpegtodfa}
For any LPEG ${\GL}$ there exists a DFA ${\D}$ such that $L({\GL}) = L({\D})$.
\end{theorem}
\begin{proof}
By Theorem \ref{theo:lpegtobfa}, LPEGs can be converted into BFAs.
BFAs can be converted into DFAs.
\end{proof}

\subsection{From a DFA to an LPEG}
\label{sec:fdtl}
An arbitrary regular expression can be converted into a PEG\cite{Medeiros20143}\cite{Oikawa_convertingregexes}.
In this section, we say that for any DFA ${\D}$ there exists an LPEG ${\GL}$ such that $L(\D) = L({\GL})$.
To prove this, we show that a PEG converted from a regular expression by \cite{Medeiros20143} is an LPEG, since DFAs can be converted into equivalent regular expressions\cite{Hopcroft:2006:IAT:1196416}.

Medeiros et al. studied the conversion and they showed the conversion function as a function $\Pi$\cite{Medeiros20143}.
The definition of the function $\Pi$ is shown in Definition \ref{def:continuation_function}.
The function $\Pi({\re},{\GP})$ takes a regular expression $\re$ and a continuation grammar ${\GP} = (N_{\GP},\Sigma,P_{\GP},\pe_s)$, and returns a PEG.
The continuation grammar is defined by a PEG ${\GP_0} = (\{\},\Sigma,\{\},\Pempty)$ for the first application.
\begin{definition} {\bf ( in \cite{Medeiros20143})}
\label{def:continuation_function}
\begin{eqnarray*}
	\Pi(\epsilon,{\GP}) &=& {\GP}\\
	\Pi(a,{\GP}) &=& (N_{\GP},\Sigma,P_{\GP},a{\pe_s})\\
	\Pi({\re}_1 {\re}_2,{\GP}) &=& \Pi({\re}_1,\Pi({\re}_2,{\GP}))\\
	\Pi({\re}_1 \mid {\re}_2,{\GP}) &=& ({N_{\GP}}'',\Sigma,{P_{\GP}}'',{\pe_s}'\Por{\pe_s}'')\\
						    &where&  ({N_{\GP}}'',\Sigma,{P_{\GP}}'',{\pe_s}'') = \Pi(\re_{2},({N_{\GP}}',\Sigma,{P_{\GP}}',\pe_{s}))\\
						    &and& ({N_{\GP}}',\Sigma,{P_{\GP}}',{\pe_s}') = \Pi(\re_{1},{\GP})\\
	\Pi({\re}^{\star},{\GP}) &=& ({N_{\GP}}',\Sigma,{P_{\GP}}'\cup\{A\leftarrow {\pe_s}'\Por\pe_{s}\},A)\ with\ A \notin N_{\GP}\\
	   			  &and& ({N_{\GP}}',\Sigma,{P_{\GP}}',{\pe_s}') = \Pi(\re,(N_{\GP}\cup\{A\},\Sigma,P_{\GP},A))
\end{eqnarray*}
\end{definition}

\begin{theorem}
\label{theo:preregexlpeg}
Let $\re$ be a regular expression and $\Pi(\re,{\GP_0}) = {\GP}$.
The PEG ${\GP}$ is an LPEG.
\end{theorem}
\begin{proof}
We assume that if $\GP$ is an LPEG, then $\Pi(\re,{\GP})$ is also an LPEG.
For any regular expression $\re$, we check whether the assumption is correct.
If so, $\Pi(\re,{\GP_0})$ is an LPEG since $\GP_0$ is obviously an LPEG.

\begin{description}
	\item[1.]\mbox{Case $\re = \epsilon$}\\
	By induction hypothesis, $\GP$ is an LPEG.
	\item[2.]\mbox{Case $\re = a$}\\
	By induction hypothesis, $\pe_s$ is a linear parsing expression.
	Since $a\pe_s = \nfree\pe$, $(N_{\GP},\Sigma,P_{\GP},a\pe_s)$ is an LPEG.
	\item[3.]\mbox{Case $\re = \re_1\re_2$}\\
	Since $\GP$ is an LPEG, $\Pi(\re_2,\GP)$ is an LPEG.
	Therefore, $\Pi(\re_1,\Pi(\re_2,\GP))$ is also an LPEG.	
	\item[4.]\mbox{Case $\re = \re_1\mid\re_2$}\\
	$\Pi(\re_1,\GP)$ is an LPEG.
	Since $\pe_s$ is a linear parsing expression, $\Pi(\re_2,({N_{\GP}}',\Sigma,\\{P_{\GP}}',\pe_{s}))$ is also an LPEG.
	Therefore, ${\pe_s}'$ and ${\pe_s}''$ are a linear parsing expression.
	Since ${\pe_s}'\Por{\pe_s}'' = \pe\Por\pe$, $({N_{\GP}}'',\Sigma,{P_{\GP}}'',{\pe_s}'\Por{\pe_s}'')$ is an LPEG.
	\item[5.]\mbox{Case $\re = \re^{\star}$}\\
	Since a nonterminal $A ( A = \nfree A)$ is a linear parsing expression, $(N_{\GP}\cup\{A\},\Sigma,\\P_{\GP},A)$ is an LPEG and $\Pi(r,(N_{\GP}\cup\{A\},\Sigma,P_{\GP},A))$ is also an LPEG.
	Since ${\pe_s}'\Por{\pe_s} = \pe\Por\pe$, $({N_{\GP}}',\Sigma,P_{\GP}'\cup\{A\leftarrow {\pe_s}'\Por\pe_{s}\},A)$ is an LPEG.
\end{description}

Hence, $\Pi(\re,{\GP_0})$ is an LPEG.

\end{proof}

\begin{theorem}
\label{theo:regexlpeg}
For any DFA ${\D}$ there exists an LPEG ${\GL}$ such that $L({\D}) = L({\GL})$.
\end{theorem}
\begin{proof}
A DFA ${\D}$ can be converted into a regular expression $\re$.
By Theorem \ref{theo:preregexlpeg}, $\re$ can be converted into an LPEG.
\end{proof}

Consequently, we derive the following theorem.

\begin{theorem}
\label{theo:lpegeqregex}
LPEGs are a class that is equivalent to DFAs.
\end{theorem}
\begin{proof}
By Theorem \ref{theo:lpegtodfa}, for any LPEG ${\GL}$ there exists a DFA ${\D}$ such that $L(\GL) = L(\D)$.
In addition, by Theorem \ref{theo:regexlpeg}, for any DFA ${\D}$ there exists an LPEG $\GL$ such that $L(\D) = L(\GL)$.
Hence, LPEGs are a class that is equivalent to DFAs.
\end{proof}

\section{Related Work}
Noam Chomsky proposed a hierarchy of formal language in order to formalize English grammar\cite{citeulike:158531}\cite{CHOMSKY1959137}.
All types of grammars in the hierarchy are currently used to describe programming language syntax.

Birman and Ullman showed formalism of recognition schemes as TS and gTS\cite{Birman:1970:TRS:905340}\cite{BIRMAN19731}.
TS and gTS were introduced in \cite{Aho:1972:TPT:578789} as TDPL and GTDPL, respectively.
A PEG is a development of GTDPL and can recognize highly nested languages such as $\{a^n b^n c^n \mid n > 0$, which is not possible in a CFG.
In this paper, we showed a subclass of PEGs that is equivalent to DFAs, which would lead to more optimized PEG-based parser generator such as \cite{Kuramitsu:2016:NPO:2986012.2986019}\cite{Grimm:2006:BET:1133255.1133987}.

Morihata showed a translation of regular expression with positive and negative lookaheads into finite state automata\cite{REwLAtoDFA}.
He used a {\em boolean finite automata} (BFAs)\cite{BRZOZOWSKI198019}, that is, {\em alternating finite automata}\cite{Chandra:1981:ALT:322234.322243}\cite{doi:10.1080/00207169008803893}, to represent positive and negative lookaheads of regular expressions as finite automata.
In this paper, we showed a translation from LPEGs to DFA and the translation is based on the Morihata's translation.

\section{Conclusion}
In this study, we formalized a subclass of PEGs that is equivalent to DFAs.
In the process of proving the equivalence of the class and DFAs, we showed the conversion from LPEGs into BFAs.
Since BFAs can be converted into DFAs, we can convert these LPEGs into DFAs.

One of our motivations is to achieve speed up of runtime by processing a part of a PEG such that the the part is regular by using DFAs.
To achieve this, we have to check whether the part of a PEG is regular.
However, this is undecidable.
On the other hand, it is decidable whether a PEG is an LPEG.
Thus, we can check whether the part of a PEG is an LPEG and convert the part into DFAs.
Since DFAs eliminate backtracking, it would lead to further optimizations of the parser generator.

As a future study, we aim to propose an algorithm for detecting a part of a PEG such that backtracking becomes necessary.

\bibliographystyle{fundam}
\bibliography{main}
\end{document}